\documentclass[10pt, journal]{IEEEtran}

\usepackage{color}
\usepackage{amsmath,epsfig,graphicx,amsthm,amssymb,url}
\usepackage{subcaption, mwe,lipsum}
\usepackage[font=footnotesize,labelfont=bf]{caption}
\usepackage{mathtools, algorithmic,algorithm,dsfont}
\usepackage[T1]{fontenc}
\usepackage[utf8]{inputenc}

\newtheorem{theorem}{Theorem}

\newtheorem{lemma}{Lemma}
\newtheorem{exmp}{Example}
\newlength{\figurewidth}
\newlength{\smallfigurewidth}
\renewcommand{\thesection}{\Roman{section}} 
\renewcommand\thesubsection{\Alph{subsection}.}

\newcommand{\hanwei}[1]{{\color{black}{{#1}}}}
\newcommand\blfootnote[1]{%
	\begingroup
	\renewcommand\thefootnote{}\footnote{#1}%
	\addtocounter{footnote}{-1}%
	\endgroup
}

\begin{document}
	\graphicspath{{figures/}}	
	\title{Computing Similarity Queries for \\ Correlated Gaussian Sources}
	\author{
		\IEEEauthorblockN{Hanwei Wu, Qiwen Wang, and Markus Flierl}\\
		\IEEEauthorblockA{School of Electrical Engineering and Computer Science}\\
		\IEEEauthorblockA{KTH Royal Insititute of Technology}\\
		\IEEEauthorblockA{Stockholm, Sweden}\\
		\IEEEauthorblockA{Email: \{hanwei, qiwenw, mflierl\}@kth.se}
	}
	\maketitle
	\frenchspacing
	\begin{abstract}
Among many current data processing systems, the objectives are often not the reproduction of data, but to compute some answers based on the data resulting from queries. The similarity identification task is to identify the items in a database that are similar to a given query item for a given metric. The problem of compression for similarity identification has been studied in \cite{amir:14:isit}. Unlike classical compression problems, the focus is not on reconstructing the original data. Instead, the compression rate is determined by the desired reliability of the answers. Specifically, the information measure \emph{identification rate} characterizes the minimum rate that can be achieved among all schemes which guarantee reliable answers with respect to a given similarity threshold. In this paper, we propose a component-based model for computing correlated similarity queries. The correlated signals are first decorrelated by the Karhunen-Lo\`{e}ve transform (KLT). Then, the decorrelated signal is processed by a distinct D-admissible system for each component. We show that the component-based model equipped with KLT  can perfectly represent the multivariate Gaussian similarity queries when optimal rate-similarity allocation applies. Hence, we can derive the \emph{identification rate} of the multivariate Gaussian signals based on the component-based model. We then extend the result to general Gaussian sources with memory. We also study the models equipped with practical component systems. We use TC-$\triangle$ schemes that use type covering signatures and triangle-inequality decision rules \cite{amir:14:isit} as our component systems. We propose an iterative method to numerically approximate the minimum achievable rate of the TC-$\triangle$ scheme. We show that our component-based model equipped with TC-$\triangle$ schemes can achieve better performance than the TC-$\triangle$ scheme unaided on handling the multivariate Gaussian sources.
	\end{abstract}
	\vspace{5ex}
	\section{Introduction} 
	\label{sec:SET}
	\blfootnote{Part of the content of Section III is submitted to 2018 IEEE Asilomar Conference on Signals, Systems, and Computers.}The problem of efficient identification and data retrieval from large databases has become more relevant in recent years. Similarity identification requires that a database returns all data items which are similar to a given query \hanwei{under a similarity threshold specified by the problem}. The notion of similarity is often defined by a specific metric measure, such as the Euclidean distance or the Hamming distance. It is required that false negative errors are not permitted in the retrieval process as they cannot be detected by further processing. This is important for some applications, such as security cameras and criminal forensic databases. On the other hand, although false positive errors can be detected by further verification, they increase the computational cost on the server side, and hence, reduce efficiency. Therefore, the tradeoff between the compression rate and the reliability of the answers to a given query is of interest.
	
	The problem of similarity identification of compressed data was first studied in \cite{ahlswede:99:it} from an information-theoretic viewpoint. In \hanwei{\cite{ahlswede:99:it}}, both false positive and false negative errors are allowed, as long as the error probability vanishes with the data block-length. Our setting though is closely related to the problem of compression for similarity queries as introduced in \cite{Amir:15:it}, \cite{amir13}. In \cite{Amir:15:it}, \cite{amir13} and this work, {\itshape false negative} errors are not permitted. \cite{Amir:15:it}, \cite{amir13} study the problem from an information-theoretic viewpoint and introduce the term {\itshape identification rate}. It characterizes the minimum compression rate that \hanwei{guarantees} query answers with a vanishing false positive probability, while false negative errors are not allowed. \cite{Amir:15:it}, \cite{amir13} provide the identification rate for Gaussian sources with quadratic distortion and for binary sources with \hanwei{the} Hamming distance. \hanwei{In \cite{Amir:15:it}, it is also proved that}, similar to the classical compression, the Gaussian source requires the largest compression rate among sources with the same variance.
	
	Since it is common to encounter correlated data in the real world, it is of interest to investigate similarity identification schemes for correlated sources. \cite{OCHOA:dcc:14} uses lossy compression as a building block to construct the TC-$\triangle$ (Type Covering signatures and triangle-inequality decision rule) scheme and the LC-$\triangle$ (Lossy Compression signatures and triangle-inequality decision rule) scheme. The LC-$\triangle$ scheme only optimizes the quantization distortion and can be achieved by employing a rate-distortion code on the triangle-inequality principle. The TC-$\triangle$ scheme is an improved version of the LC-$\triangle$ scheme by optimizing jointly the quantization distortion and the expected query codeword distance. The results in \cite{OCHOA:dcc:14} show that the compression rate of TC-$\triangle$ can achieve the identification rate for the case with binary sources and the Hamming distance. 
	
	In \cite{steiner:16:it}, the authors present a shape-gain quantizer for i.i.d. Gaussian sequences: scalar quantization is applied to the magnitude of the data vector. The shape (the projection on the unit sphere) is quantized using a wrapped spherical code \cite{hamkins:97:it}. \cite{hanwei:dcc:17} proposes tree-structured vector quantizers that hierarchically cluster the data using $k$-center clustering. In \cite{hanwei:17:asilomar}, the authors compare two transform-based similarity identification schemes to cope with exponentially growing codebooks for high-dimensional data. One of the proposed schemes, that is, the component-based approach, exhibits both good performance and low search complexity. However, the theoretical analysis for the component-based setting is still an open problem and remains to be investigated. Besides, for correlated sources, no analytical results on the minimum achievable rates of the above schemes are provided.
	
	In this paper, \hanwei{we first propose a component-based model for computing the identification rate for multivariate Gaussian sources. We use the Karhunen-Lo\`{e}ve transform (KLT) to create independent $D$-admissible component systems. We show that the component-based model equipped with KLT can perfectly represent the multivariate Gaussian similarity queries when optimal rate-similarity allocation applies. We then extend the result to the identification rate of the general Gaussian sources with memory. To evaluate the performance of practical schemes, we replace the optimal component system with the state-of-the-art TC-$\triangle$ schemes. We propose an iterative method inspired by the Blahut–Arimoto (BA) \cite{blahut:72:it}, \cite{arimoto:72:it} and related algorithms \cite{walsh:15:sp} to numerically approximate the minimum achievable rate of TC-$\triangle$ schemes. The simulation shows that our component-based model equipped with TC-$\triangle$ schemes has better performance than the TC-$\triangle$ scheme unaided on handling the multivariate Gaussian sources.}
		
	The outline of this paper is as follows: In Section $2$, we give a brief description of the problem's background and key concepts \footnote{We follow the problem setup and adopt most notations in \cite{Amir:15:it} and \cite{amir13}. Therefore, we refer to \cite{Amir:15:it} and \cite{amir13}} for more detailed background and problem description. \hanwei{In Section $3$, we propose our component-based model for computing the identification rate of multivariate Gaussian sources. In Section $4$, we extend the identification rate result for general Gaussian sources with memory.} In Section $5$, we propose an iterative method to approximate the \hanwei{minimum achievable rate} of the TC-$\triangle$ scheme. Then we compare the TC-$\triangle$ scheme with the component-based scheme for i.i.d. and multivariate Gaussian sources. The conclusions are given in Section $6$. 
	
		{\itshape The notational conventions in this work are as follows. Uppercase nonboldface symbols such as $X$ are used to denote random variables; and lowercase nonboldface symbols such as $x$ are used to denote sample values of those random variables. Vectors and matrices of random variables or their sample values are denoted by boldface symbols. For example, $\mathbf{X}$ and $\mathbf{x}$ are vectors (or sometimes matrices from the context) of random variables $X$ and its sample values $x$, respectively. The $i$th entry of a vector $\mathbf{X}$ is denoted by $X_i$.}
	\section{Quadratic Similarity Queries}
	Let $\mathbf{y} = (y_1,y_2,..,y_n)^T$ denote the query sequence and $\mathbf{x} = (x_1,x_2,..,x_n)^T$ the data sequence. A rate-\hanwei{$R_{\text{ID}}$} identification scheme $(T, g)$ consists of a \emph{signature assignment} function: $T: \mathbb{R}^n\rightarrow \{1,2,\cdots, 2^{nR_{\text{ID}}}\}$, and a {\itshape query function} $g$:$\{1,2,\cdots, 2^{nR_{\text{ID}}}\} \times \mathbb{R}^n \rightarrow\{\text{no}, \text{maybe}\}$. The database keeps only a short signature $T(\mathbf{x})$ for each $\mathbf{x}$. And the output decision $\emph{maybe}$ or $\emph{no}$ of a query function indicates whether $\mathbf{x}$ and $\mathbf{y}$ are probably $D_{\text{ID}}$-similar or not. The sequences $\mathbf{x}$ and $\mathbf{y}$ are called $D_{\text{ID}}$-$\emph{similar}$ if $d(\mathbf{x}, \mathbf{y}) \leq D_{\text{ID}}$, where  \hanwei{we restrict our consideration to \emph{additive distortion measures}}
	\begin{equation}
		d(\mathbf{x}, \mathbf{y}) \triangleq \frac{1}{n}\sum_{i = 1}^n \rho(x_i, y_i),
	\end{equation}
	$\rho$ is an arbitrary per-letter distance measure \hanwei{specified by the problem}, and $D_{\text{ID}}$ is the {\itshape similarity threshold}. Specifically, the quadratic similarity is 
		\begin{equation}
			d(\mathbf{x}, \mathbf{y}) \triangleq \frac{1}{n}\|\mathbf{x}-\mathbf{y}\|^2 =  \frac{1}{n}\sum_{i = 1}^n (x_i - y_i)^2,
	\end{equation}
	where $\|\cdot\|$ is the standard Euclidean norm.
	
	A similarity query retrieves all data items that are $D_{\text{ID}}$-similar \hanwei{to the query sequence}. A scheme is called $D_{\text{ID}}$-{\itshape admissible} if we obtain $g( T(\mathbf{x}), \mathbf{y} )=$ maybe for any pair of data item and query $(\mathbf{x}, \mathbf{y})$ which is $D_{\text{ID}}$-$\emph{similar}$. 
	
	Now, consider a probabilistic model for database and query. The objective is to design $D_{\text{ID}}$-admissible schemes that minimize the probability of the output \emph{maybe} for given distributions of database vectors $\mathbf{X}$ and query vectors $\mathbf{Y}$. According to \cite{Amir:15:it}, for a $D_{\text{ID}}$-admissible scheme, this probability is calculated as
	\begin{equation}
		\label{equ:1}
		\begin{split}
			\hspace*{-15pt}&\Pr \left\lbrace g( T(\mathbf{X}), \mathbf{Y} ) = \text{maybe} \right\rbrace\\
			&= \Pr \left\lbrace g(T(\mathbf{X}),\mathbf{Y}) = \text{maybe} | d(\mathbf{X,Y})\leq D_{\text{ID}}\right\rbrace \\
			&\hspace*{12pt}\Pr \left\lbrace d(\mathbf{X,Y})\leq D_{\text{ID}} \right\rbrace \\
			&+ \Pr \left\lbrace g(T(\mathbf{X}),\mathbf{Y}) = \text{maybe} , d(\mathbf{X,Y})> D_{\text{ID}}\right\rbrace\\
			&= \Pr \left\lbrace d(\mathbf{X,Y})\leq D_{\text{ID}}\right\rbrace + \Pr(\varepsilon),
		\end{split}
	\end{equation}
	where the second equality follows from $ \Pr \left\lbrace g( T(\mathbf{X}), \mathbf{Y} ) = \text{maybe} | d(\mathbf{X,Y})\leq D_{\text{ID}}\right\rbrace =1$ by the requirement of $D_{\text{ID}}$-admissibility. Hence, minimizing (\ref{equ:1}) is equivalent to minimizing the probability of false positives $\Pr(\varepsilon)$. That is, the probability $\Pr \left\lbrace g(T(\mathbf{X}), \mathbf{Y}) = \text{maybe}\right\rbrace$ can be used as a performance measure for the investigated schemes. In the following, we use the abbreviation $\Pr \left\lbrace\text{maybe}\right\rbrace$ for the probability that a scheme outputs maybe.
	
	For given distributions $P_X$ and $P_Y$ and a similarity threshold $D_{\text{ID}}$, a rate $R$ is said to be $D_{\text{ID}}$-{\itshape achievable} if there exist a sequence of $D_{\text{ID}}$-admissible schemes $\left(T^{(n)}, g^{(n)}\right)$ that can achieve a vanishing $\Pr\{\text{maybe}\}$ as $n$ approaches infinity: 
	\begin{equation}
		\underset{n\rightarrow \infty}{\lim} \Pr\{g^{(n)}(T^{(n)}(\mathbf{X}),\mathbf{Y}) = \text{maybe}\} = 0.
	\end{equation}
	The {\itshape identification rate} $R_{\text{ID}}^*$ of the source is defined as the infimum of all $D_{\text{ID}}$-{\itshape achievable} rates.
	\hanwei{
	\section{Identification Rate of \\ Multivariate Gaussian Sources} 
	\label{sec:MWG}
	\subsection{A Component-based Model}
	 We propose a component-based model to compute the identification rate for multivariate Gaussian sources. The idea is that the input which consists of $M$th order multivariate Gaussian signals is first decorrelated into $M$ components by the Karhunen-Lo\`{e}ve transform (KLT) for further processing. After the transform, we use a $D_{\text{ID}}^{(m)}$-admissible system for each component and they together form an $M$-component system. The $m$th component system answers $\emph{maybe}$ if the transformed $m$th query-database pair $\left(\mathbf{x}^{(m)}, \mathbf{y}^{(m)}\right)$ satisfies $d\left(\mathbf{x}^{(m)}, \mathbf{y}^{(m)}\right) \leq D_{\text{ID}}^{(m)}$.
	 
	 We consider an $M$th order zero-mean stationary Gaussian process 
	 \begin{equation}
	 \label{eq:multigaussian}
	 f_{\mathbf{\tilde{X}}}(\mathbf{\tilde{x}}) = \frac{1}{(2\pi)^{N/2}|\mathbf{C}_M|^{1/2} e^{-\frac{1}{2}\mathbf{\tilde{x}}^T\mathbf{C}_M^{-1}\mathbf{\tilde{x}}}},
	 \end{equation}
	 where $\mathbf{\tilde{x}}$ is a vector of $M$ consecutive samples and $\mathbf{C}_M$ is the $M$th order autocovariance matrix.  Since $\mathbf{C}_M$ is a real symmetric matrix, it has the eigendecomposition as
	 \begin{equation}
	 \mathbf{C}_M = \mathbf{A}_M\mathbf{\Sigma}_M\mathbf{A}_M^T,
	 \end{equation}
	 where $\mathbf{A}_M = \left[\mathbf{e}_M^{(1)}, \mathbf{e}_M^{(2)},\cdots, \mathbf{e}_M^{(M)}\right]$ is the matrix whose columns are the eigenvectors of $\mathbf{C}_M$, $\mathbf{\Sigma}_M$ is a diagonal matrix with eigenvalues $\xi_M^{(m)}$ as diagonal entries. The source can be decorrelated by the transform $\mathbf{x} = \mathbf{A}_M^T\mathbf{\tilde{x}}$, that is
	 \begin{equation}
	 \mathbb{E}[\mathbf{X}\mathbf{X}^T] = \mathbf{A}_M^T\mathbb{E}[\mathbf{X}\mathbf{X}^T]\mathbf{A}_M = \mathbf{A}_M^T \mathbf{C}_M\mathbf{A}_M = \mathbf{\Sigma}_M.
	 \end{equation}
	 We denote $\mathbf{Q}$ as the transform we use for the component-based model such as $\mathbf{x} = \mathbf{Q}\mathbf{\tilde{x}}$, where $\mathbf{\tilde{x}}$ is the input signal.
	 Then the transpose of the eigenmatrix $\mathbf{A}_M^T$ is the Karhunen-Lo\`{e}ve transform (KLT).

      Let $(T^*, g^*)$ denote an optimal identification system for multivariate Gaussian sources, that is, $(T^*, g^*)$ can achieve the identification rate of multivariate Gaussian sources.  In next two sections, we derive conditions that preserve the $D_{\text{ID}}$-admissibility and $D_{\text{ID}}$-achievability of the optimal identification system  $(T^*, g^*)$ for the $M$-component model. 
	\subsection{$D_{\text{ID}}$-admissible Condition}
	Maintaining the $D_{\text{ID}}$-admissibility after the transform requires that the similarity measure of the origin domain is persevered in the transform domain. Since $\mathbf{Q} = \mathbf{A}_M^T$ is an orthogonal matrix, the KLT  is an orthogonal transform $\mathbf{Q}^T\mathbf{Q} = \mathbf{I}$ and preserves the quadratic distance
	\begin{align}
	d(\mathbf{x}, \mathbf{y}) &= \frac{1}{M}(\mathbf{x}-\mathbf{y})^T(\mathbf{x}-\mathbf{y}) \\
	&=\frac{1}{M}(\mathbf{Q}\mathbf{\tilde{x}}-\mathbf{Q}\mathbf{\tilde{y}})^T(\mathbf{Q}\mathbf{\tilde{x}}-\mathbf{Q}\mathbf{\tilde{y}}) \\
	&= \frac{1}{M}(\mathbf{\tilde{x}}-\mathbf{\tilde{y}})^T\mathbf{Q}^T\mathbf{Q}(\mathbf{\tilde{x}}-\mathbf{\tilde{y}}) \\
	&= \frac{1}{M}(\mathbf{\tilde{x}}-\mathbf{\tilde{y}})^T(\mathbf{\tilde{x}}-\mathbf{\tilde{y}})  \\
	&=d(\mathbf{\tilde{x}}, \mathbf{\tilde{y}}).
	\end{align}
	
	In order to preserve the $D_{\text{ID}}$-admissibility of the optimal system $(T^*, g^*)$, the similarity threshold for an $M$th order component-based model $\frac{1}{M} \sum_{m = 1}^M D_{\text{ID}}^{(m)}$ should be at least the same as the given similarity threshold $D_{\text{ID}}$ 
	\begin{align}
	\label{eq: Did-admissible}
	d(\mathbf{x},\mathbf{y}) \leq D_{\text{ID}} \leq \frac{1}{M} \sum_{m = 1}^M D_{\text{ID}}^{(m)}.
	\end{align}
	\subsection{$D_{\text{ID}}$-achievable Condition}
	Let $(T_m, g_m)$ denote the identification system for the $m$th component. Lemma \ref{lemma: dachievability} shows the conditions of achieving a vanishing $\Pr\{\text{maybe}\}$ of the component-based model for multivariate Gaussian signals.
\begin{lemma}
	\label{lemma: dachievability}
	Consider data sequence $\mathbf{X}$ and query sequence $\mathbf{Y}$ both being concatenations of $N$ independent blocks of  zero-mean multivariate Gaussian random variables with blocklength $M$ for $D_{\text{ID}}$-similarity identification, where $n = MN$ is the length of the whole sequence. We have a vanishing $\Pr\{\text{maybe}\}$ for the overall system
\begin{equation}
\underset{n \rightarrow \infty}{\lim}\Pr\{g^{(n)}(T^{(n)}(\mathbf{X}), \mathbf{Y})= \text{maybe}\} = 0
\end{equation}
if and only if, 
\begin{equation}
\exists m, \hspace{0.2em}
\underset{N \rightarrow \infty}{\lim}\Pr\{g_m^{(N)}\left(T^{(N)}_m\left(\mathbf{X}^{(m)}\right), \mathbf{Y}^{(m)}\right) = \text{maybe} \}= 0.
\end{equation}
\end{lemma}
\begin{proof}
As shown in \cite{Amir:15:it}, the $\Pr\{\text{maybe}\}$ can be bounded from the above by
\begin{align}
\label{eq: pmaybeanalys}
\Pr\{\text{maybe}\} \leq &\Pr\{\text{maybe}| \mathbf{X}\in S_X^{\text{typ}}, \mathbf{Y} \in S_Y^{\text{typ}}\} \nonumber\\
&+ \Pr\{\mathbf{X} \not\in S_X^{\text{typ}}\} + \Pr\{\mathbf{Y} \not\in S_Y^{\text{typ}}\},
\end{align}
where  $S_X^{\text{typ}}$ and $S_Y^{\text{typ}}$ are the \emph{typical spheres}. Since $\Pr\{\mathbf{X} \not\in S_X^{\text{typ}}\}$ and $\Pr\{\mathbf{Y} \not\in S_Y^{\text{typ}}\}$ vanishes with $n$, we focus on the first term of (\ref{eq: pmaybeanalys}). 

Recall that the input multivariate Gaussian signals are first decorrelated by the KLT. Furthermore, the uncorrelatedness of jointly distributed Gaussian random variables imply independence, thus, the transform that decorrelates the multivariate Gaussian sources can also create independent components $X^{(m)}$. Due to the independence of the components, we can write 
\begin{align}
&\Pr \left\{g (T(\mathbf{X}), \mathbf{Y}) = \text{maybe}\right\} \\
&\propto \Pr\left\{d(T^{-1}(T(\mathbf{X})), \mathbf{Y})\leq D_{\text{ID}} \right \} \\
\label{eq:Dlarger}
&\leq \Pr \left\{\frac{1}{M}\sum_{m = 1}^{M}d\left(T_m^{-1}\left(T_m\left(\mathbf{X}^{(m)}\right)\right), \mathbf{Y}^{(m)}\right) \right.\\
&\hspace{1em} \left. \leq \frac{1}{M}\sum_{m = 1}^{M}D_{\text{ID}}^{(m)}\right\} \nonumber\\
&=\Pr \left\{ d\left(T_1^{-1}\left(T_1\left(\mathbf{X}^{(1)}\right)\right), \mathbf{Y}^{(1)}\right) \leq D_{\text{ID}}^{(1)},\cdots, \right. \\
&\hspace{1em}\left. d\left(T_M^{-1}\left(T_M\left(\mathbf{X}^{(M)}\right)\right), \mathbf{Y}^{(M)}\right)\leq D_{\text{ID}}^{(M)}\right\} \nonumber \\
\label{productproperty}
&=\prod_{m = 1}^{M} \Pr\left\{d\left(T_m^{-1}\left(T_m\left(\mathbf{X}^{(m)}\right)\right), \mathbf{Y}^{(m)}\right) \leq D_{\text{ID}}^{(m)}\right\},
\end{align}
	where $T^{-1}(k) \triangleq \{\mathbf{x}: T(\mathbf{x}) = k\}$ represents the set of vectors that have the same signature. (\ref{eq:Dlarger}) follows from that quadratic distance $d(\cdot)$ is an additive distortion measure and the $D$-admissible condition (\ref{eq: Did-admissible}). (\ref{productproperty}) follows because the joint probability of independent events equals the product of their probabilities. 
	
Therefore, the $\Pr\{\text{maybe}\}$ of the overall system is proportional to the product of the components' $\Pr\{\text{maybe}\}$
	\begin{align}
		\label{eq: productproperty}
		&\Pr \left\{g (T(\mathbf{X}), \mathbf{Y}) =\text{maybe}\right\}  \\ \nonumber
		&\propto \prod_{m = 1}^{M} \Pr \left\{g_m\left(T_m\left(X^{(m)}\right), Y^{(m)}\right) =\text{maybe}\right\}.
	\end{align}
	
	Let the blocklength $M$ goes to infinity, as a result, the overall length of the sequence $n$ also tends to infinity. In order to have a vanishing $\Pr \left\{g (T(\mathbf{X}), \mathbf{Y}) =\text{maybe}\right\}$ for the overall system, there must exist some components $m$ such that its $\Pr \left\{g_m\left(T_m\left(X^{(m)}\right), Y^{(m)}\right) =\text{maybe}\right\}$ goes to zero.
\end{proof}
Due to the product property of the $\Pr\{\text{maybe}\}$ (\ref{eq: productproperty}), the database vectors are labeled as \emph{maybe} if and only if all of its component systems are determined as \emph{maybe}. Therefore, the final output of the component-based model can be achieved by the logic AND decision. Fig. \ref{fig:qiwencheme} shows the proposed component-based model. 
\begin{figure*}[t!]
	\centering
	\includegraphics[width=0.8\textwidth]{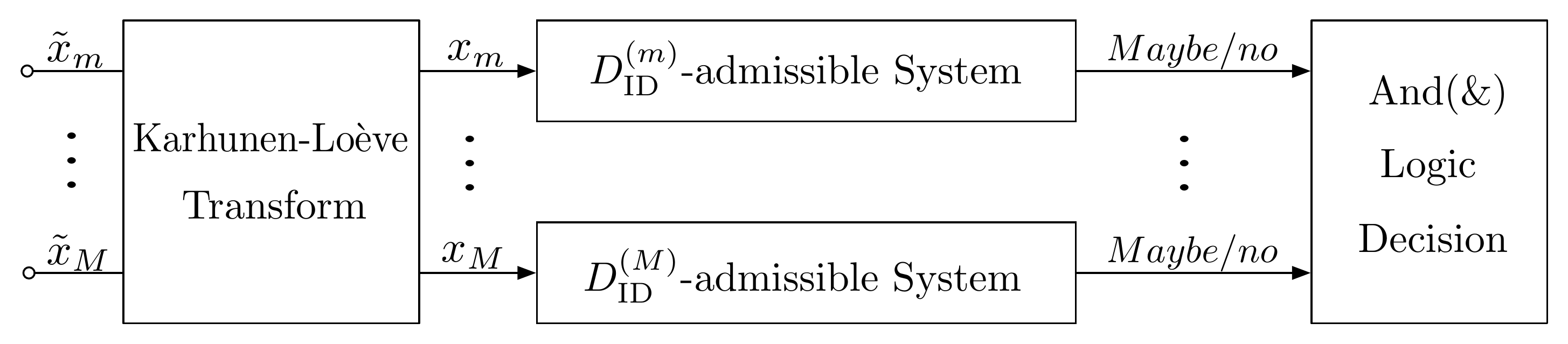}
	\caption{Component-based model for similarity identification.}
	\label{fig:qiwencheme}
\end{figure*}
	\subsection{Identification Rate $R_{\text{ID}}^{M*}$} 
	We define the \emph{identification rate} of a multivariate Gaussian source $R_{\text{ID}}^{M*}$ as the infimum of all $D_{\text{ID}}$-achievable rates. In previous sections, we show that we can use KLT to create independent component systems. We also derive the $D_{\text{ID}}$-admissible and $D_{\text{ID}}$-achievable conditions for the component-based model. In the proof of Theorem \ref{theorem:2}, we formulate the problem of computing the $R_{\text{ID}}^{M*}$ as a rate-similarity optimization problem under the $D_{\text{ID}}$-admissible and $D_{\text{ID}}$-achievable conditions of the component-based model. We show that under the optimal rate-similarity allocation, the $D_{\text{ID}}$-admissible and $D_{\text{ID}}$-achievable conditions for the enforced component-based model become equivalent to the corresponding conditions of the optimal system $(T^*, g^*)$. As a result, we can conclude that the identification rate $R_{\text{ID}}^{M*}$ can be achieved by the component-based model with optimal rate-similarity allocation. Note that we consider the case that the query and the database follow the same multivariate Gaussian distribution so that the query can be decorrelated by using the same KLT as used for the database. 
	\begin{theorem}
		\label{theorem:2}
		The identification rate of $M$th order multivariate Gaussian sources (\ref{eq:multigaussian}) is 
		\begin{equation}
		\centering
		\label{eq:mgs}
		R_{\text{ID}}^{M*} = \frac{1}{M}\sum_{m = 1}^{M}\max\left(0, \log \frac{\xi^{(m)}}{\tau}\right),
		\end{equation}
		\begin{equation}
		\centering
		\label{eq:DIDforcom}
		D_{\text{ID}} = \frac{1}{M}\sum_{m = 1}^{M}\max\left(0, 2(\xi^{(m)} - \tau)\right)
		\end{equation}
		with $\tau\in \left(0, \xi_{\max}\right]$. $\xi^{(m)}$ is the $m$th eigenvalue of the autocovariance matrix $\mathbf{C}_M$ and  $\xi_{\max}$ is its largest eigenvalue.
		
		The identification rate $R_{\text{ID}}^{M*}$ approaches infinity when the similarity threshold is
		\begin{equation}
			D_{\text{ID}} \geq \frac{2}{M}\overset{M}{\underset{m = 1}{\sum}}\xi^{(m)}.
		\end{equation}
	\end{theorem}
	\begin{proof}
		  Since the KLT is an invertible transform, the mutual information in the transform domain is preserved \cite{kraskov2004estimating}. Thus, the achievable rate required by the component-based model is identical to the rate for the original signal. In addition, since the components $X^{(m)}$ are independent of each other, then the $M$th order mutual information between the input data and its signature is the sum of the mutual information of signal and signature of all components
		 \begin{equation}
		 I_M(\mathbf{X};T(\mathbf{X})) = \sum_{m = 1}^M I_1\left(X^{(m)};T_m(X^{(m)}\right).
		 \end{equation}
		 where $I_M$ denotes the $M$th order mutual information. 
		 
		 From the $D_{\text{ID}}$-achievable condition of the component-based model, we know that there must exist some components that have vanishing $\Pr\{\text{maybe}\}$. Furthermore, we are interested in computing the minimum achievable rate. Hence, we assume that each component uses an ideal $D_{\text{ID}}^{(m)}$-admissible scheme such that each component system operates on its identification rate curve. We can then define the achievable rate of the component-based model as the average of the component rates $R_{\text{ID}}^{(m)}$, {\itshape i.e.,}
		 \begin{equation}
		 \label{eq: ratedefinition}
		 R_{\text{ID}} = \frac{1}{M}\sum_{m = 1}^M R_{\text{ID}}^{(m)}\left(D_{\text{ID}}^{(m)}\right) .
		 \end{equation}
		 
		 For a given similarity threshold $D_{\text{ID}}$ of the original signal, the infimum of all achievable rate of the component-based model can be obtained by solving the following constrained optimization problem
		\begin{align}
		\label{equ:lagrange1}
		\begin{split}
		&\min_{D_{\text{ID}}^{(1)}, \dots, D_{\text{ID}}^{(M)}}   R_{\text{ID}} = \frac{1}{M}\sum_{m = 1}^{M} R_{\text{ID}}^{(m)}\left(D_{\text{ID}}^{(m)}\right)\\
		&\hspace{2em}\text{s.t.} \hspace{1em}  \frac{1}{M}\sum_{m = 1}^{M}D_{\text{ID}}^{(m)} \geq D_{\text{ID}} , \\
		&\hspace{2em}\text{s.t.} \hspace{1em} D_{\text{ID}}^{(m)} \geq 0.
		\end{split}
		\end{align}
		 The first inequality constraint follows from the $D_{\text{ID}}$-admissible condition (\ref{eq: Did-admissible}) for the component-based model. The second inequality constraint follows from the nonnegativity definition of the component similarities. 
		
		Note, all identification rate functions $R^{(m)}_{\text{ID}}\left(D^{(m)}_\text{ID}\right)$ of the components are convex and strictly increasing. Hence, we consider the equivalent problem
		\begin{align}
		\label{equ:lagrange2}
		\begin{split}
		&\min \hspace{1em}  J = R_{\text{ID}}\left(D_{\text{ID}}^{(m)}\right) - vD_{\text{ID}} \\
		&\hspace{1em} \text{s.t.} \hspace{0.8em}  D_{\text{ID}}^{(m)} \geq 0,
		\end{split}
		\end{align}
		where $v$ is a positive Lagrangian multiplier. 
		
		It is shown in \cite{amir13} that the identification rate for i.i.d. Gaussian sources is 
		\begin{align}
		\label{equ:bts}
		R_{\text{ID}}^*(D_{\text{ID}}) =
		\begin{cases}
		\log(\frac{2\sigma^2}{2\sigma^2-D_{\text{ID}}}) \hspace{1em}  &\text{for} \ 0 \leq D_{\text{ID}} < 2\sigma^2\\
		\infty  &\text{for} \ D_{\text{ID}} \geq 2\sigma^2.
		\end{cases}
		\end{align}
		Since the variance $\sigma_m^2$ of the component $X^{(m)}$ is equal to the eigenvalue $\xi_M^{(m)}$ of the $M$th order autocovariance matrix $\mathbf{C}_M$, when $0 \leq D_{\text{ID}} < 2\sigma^2$, the derivative of the cost function $J$ with respect to $D_{\text{ID}}^{(m)}$ can be expressed as:
		\begin{equation}
		\label{partial1}
		\frac{\partial J}{\partial D_{\text{ID}}^{(m)}} = \frac{1}{\ln(2)\left(2\xi_M^{(m)}-D_{\text{ID}}^{(m)}\right)} -v.
		\end{equation}
		By setting (\ref{partial1}) to zero, we obtain that $D_{\text{ID}}^{(m)}$ is determined by the eigenvalue $\xi_M^{(m)}$ and the value of $v$, {\itshape i.e.,}
		\begin{equation}
		\label{pareto1}
		D_{\text{ID}}^{(m)}= 2\xi_M^{(m)}- \frac{1}{v\ln(2)}.
		\end{equation}
		
		Let $v$ be expressed as $\frac{1}{2\ln(2)\tau}$, $\tau \geq 0$. In order to satisfy the non-negative constraint of $D_{\text{ID}}^{(m)}$, each component will only be assigned with a positive similarity threshold when the value of  $\tau$ is smaller than its corresponding eigenvalue $\xi_M^{(m)}$,
		\begin{equation}
		\label{eq: Dcomponent}
			D_{\text{ID}}^{(m)} = \max\left(0, 2(\xi^{(m)} - \tau)\right).
		\end{equation}
		In order to have at least one component assigned with positive similarity threshold, we set the largest $\tau$ as $\underset{m}{\max} \ \xi_M^{(m)} = \xi_{\max}$. 
		
		By substitute (\ref{eq: Dcomponent}) into (\ref{equ:bts}), we can obtain the identification rate for the $m$th component. Then, the infimum of the achievable rate for the component-based model is 
		\begin{equation}
		\label{eq:r1}
		R_{\text{ID}}^{M*}(\tau) = \frac{1}{M}\sum_{m = 1}^{M} \max\left[0, \log\left(\frac{\xi_M^{(m)}}{\tau}\right)\right),
		\end{equation}
		and the corresponding similarity threshold of the component-based model is 
		\begin{equation}
		\label{eq:d1}
		D_{\text{ID}}(\tau) =  \frac{1}{M}\sum_{m = 1}^{M} \max\left(0, 2(\xi_M^{(m)}-\tau)\right).
		\end{equation}
		The optimal rate-similarity curve of the component-based model can be obtained by sweeping over permitted values of $\tau$.
		
		According to the Karush–Kuhn–Tucker conditions, the optimal point occurs on the constraint surface.  Therefore, the inequality constraint in (\ref{equ:lagrange1}) can reach equality when optimal point is achieved
		\begin{equation}
		D_{\text{ID}} = \frac{1}{M}\sum_{m = 1}^{M}D_{\text{ID}}^{(m)}.
		\end{equation}
		On the other hand, the optimal condition also achieves equality in (\ref{eq: productproperty}), hence, the component-based model preserves the original $\Pr\{\text{maybe}\}$. 
		Therefore, the imposed component-based model preserves the same characteristics of the original system under the optimal rate-similarity allocation. We can conclude that the derived optimal rate-similarity functions (\ref{eq:r1}), (\ref{eq:d1}) based on the component-based model are identical to the identification rate function of the multivariate Gaussian sources.  
		
		For the limit case of $\tau = 0$, the model's identification rate (\ref{eq:r1})  approaches infinity and the model's similarity threshold is 
		\begin{equation}
			D_{\text{ID}} = \frac{2}{M}\overset{M}{\underset{m = 1}{\sum}}\xi^{(m)},
		\end{equation}
		where each component similarity threshold $D_{\text{ID}}^{(m)}$ is $2\xi_M^{(m)}$. If the given model's similarity threshold $D_{\text{ID}}$ is larger than $\frac{2}{M}\overset{M}{\underset{m = 1}{\sum}}\xi^{(m)}$, there must exist components that have similarity thresholds larger than $2\xi_M^{(m)}$. Therefore, according to (\ref{equ:bts}), the overall identification rate approaches infinity.
	\end{proof}
}
	The Theorem shows that the optimal identification rate can be achieved by activating the components according to their variances after the KLT. At the lowest rate, only the component with the largest variance is activated (assigned with positive similarity threshold). In this case, the $M$component  model uses only one component. Then, as the rate increases, the remaining components are activated in the order of their component variances. The activated components operate according to the Pareto condition. 
	
	\hanwei{
	Similar to i.i.d. Gaussian sources, multivariate Gaussian sources also have a similarity threshold limit that the systems can achieve vanishing $\Pr\{\text{maybe}\}$. The similarity threshold limit for multivariate Gaussian sources is twice of the trace of its covariance matrix. If the systems are given a similarity threshold that is larger than the similarity threshold limit of the processed source, the query and database are inherently similar. Hence, the $\Pr\{\text{maybe}\}$ can never vanish regardless of what system is used.
	\section{Identification Rate of \\ Gaussian Sources with Memory} 
	\label{sec:GSWM}
	In this section, we extend the identification rate result of multivariate Gaussian sources to general Gaussian sources with memory. The proof of Theorem {\ref{theorem:2}} shows that high dimensional multivariate Gaussian similarity queries can be perfectly represented by the optimal component-based model. Based on this, we study the general Gaussian sources with memory using the optimal component-based model and apply the Szegö's theorem for sequences of Toeplitz matrices \cite{greander58} under the limiting cases. Theorem \ref{theorem:1} gives the result of the identification rate of the zero-mean stationary Gaussian process.
	\begin{theorem}
			\label{theorem:1}
		The identification rate function $R_{\text{ID}}^*(D_{\text{ID}})$ for a zero-mean stationary Gaussian process with memory is 
		\begin{equation}
			\label{eq:rwithm}
			R_{\text{ID}}^{*} = \frac{1}{2\pi}\int_{-\pi}^{{\pi}}\max\left(0, \log\left(\frac{\Phi_{XX}(\omega)}{\tau}\right)\right)d\omega,
		\end{equation}
		\begin{equation}
			\label{eq:dwithm}
			D_{\text{ID}}^{*}  = \frac{1}{2\pi}\int_{-\pi}^{{\pi}}\max\left(0, 2\left(\Phi_{XX}(\omega) -\tau\right)\right)d\omega
		\end{equation}
		with $\tau \in \left(0, M_\phi\right]$, where $\Phi_{XX}(\omega)$ is the power spectral density of the source and $M_\phi$ is the \emph{essential supremum} of $\Phi_{XX}(\omega)$.
		
		The identification rate $R_{\text{ID}}^{*}$ approaches infinity when the similarity threshold is 
		\begin{equation}
		\label{eq: signalpower}
		D_{\text{ID}} \geq \frac{1}{\pi}\int_{-\pi}^{\pi}\Phi_{XX}(\omega)d\omega.
		\end{equation}
	\end{theorem}
	\begin{proof}
		Given the stationary Gaussian source $\{\tilde{X}_n\}$,  we can decompose the source into vectors $\mathbf{\tilde{X}}$ of $M$ successive random variables and describe those vectors with a $M$th order multivariate Gaussian distribution (\ref{eq:multigaussian}). Then we can apply the KLT transform on the decomposed vectors $\mathbf{X} = \mathbf{A}_M^T\mathbf{\tilde{X}}$, where $\mathbf{A}_M$ is the eigenmatrix of the covariance matrix $\mathbf{C}_M$. The resulted decorrelated source $\{ X_n\}$ is given by the concatenation of the random vectors $\mathbf{X}$. 
		
		After the KLT, the signal is processed by the independent component $D_{\text{ID}}^{(m)}$-admissible systems. The identification rate of the $M$th order multivariate Gaussian sources is known from the Theorem \ref{theorem:2}. Therefore, we can obtain the identification rate  of the stationary Gaussian random process $\{\tilde{X}_n\}$ by taking the  limit
		\begin{equation}
		R_{\text{ID}}^{*}(D_{\text{ID}}) =  \lim_{M \rightarrow \infty} R_{\text{ID}}^{M*}(D_{\text{ID}})
		\end{equation}
         
        The autocovariance matrices $\mathbf{C}_M$ of stationary processes are Toeplitz matrices $T_M(\Phi)$, where $\Phi$ is the power spectral density of the source defined by the Fourier series of the elements on the $m$th diagonal of $\mathbf{C}_M$
		\begin{equation}
		\Phi_{XX}(\omega) = \sum_{m = -\infty}^{\infty}\phi_m e^{-j\omega m}.
		\end{equation}
		If the \emph{essential supremum} and \emph{essential infimum} of $\Phi_{XX}(\omega)$ are finite, the theorem for sequences of Toeplitz matrices \cite{greander58} states that
		\begin{equation}
		\label{eq:sequenceofmatrix}
		\lim_{M \rightarrow \infty} \frac{1}{M}\sum_{m = 1}^MG\left(\xi_M^{(m)}\right) =  \frac{1}{2\pi}\int_{-\pi}^{\pi} G\left(\Phi_{XX}(\omega)\right)d\omega,
		\end{equation}
		for any function $G$ is continous on the range of $\Phi_{XX}(\omega)$. Therefore, when $M \rightarrow \infty$, we can apply (\ref{eq:sequenceofmatrix}) to the identification rate function of multivariate Gaussian sources (\ref{eq:r1}), (\ref{eq:d1}), and obtain (\ref{eq:rwithm}), (\ref{eq:dwithm}). 
		
		In addition, according to the Lemma 4.1 of \cite{gray2006}, the eigenvalues of Toeplitz matrix $T_M(\Phi)$ are bounded by the \emph{essential infimum} and \emph{essential supremum} of $\Phi_{XX}(\omega)$. We denote the \emph{essential supremum} of $\Phi_{XX}(\omega)$ as $M_\phi$. Hence, we can rewrite the permitted values of $\tau$ as $\left(0, M_\phi\right]$.
		
		The extreme case of (\ref{eq: signalpower}) is obtained when $\tau = 0$ by following a similar arrangement of the proof of Theorem \ref{theorem:1}.
	\end{proof}
          The identification rate function follows a similar "reverse water-filling" process as the rate-distortion function of Gaussian sources with memory \cite{Cover2006}.  The value of  $\tau$ starts decreasing from $M_\Phi$, the rate is first allotted to frequencies with the largest altitudes. As the value of $\tau$ decreases, the rate is put to frequencies with lower altitudes. The difference is that the distortion is calculated as the integral of the minimum values of the frequency attitude and the water level, while the similarity threshold is calculated as the integral of the differences between the frequency attitude and the water level. An example of the reverse water-filling process is shown in Figure \ref{fig:waterfilling}.
          \begin{figure}[!ht]
          	\centering
          	\includegraphics[width=0.5\textwidth]{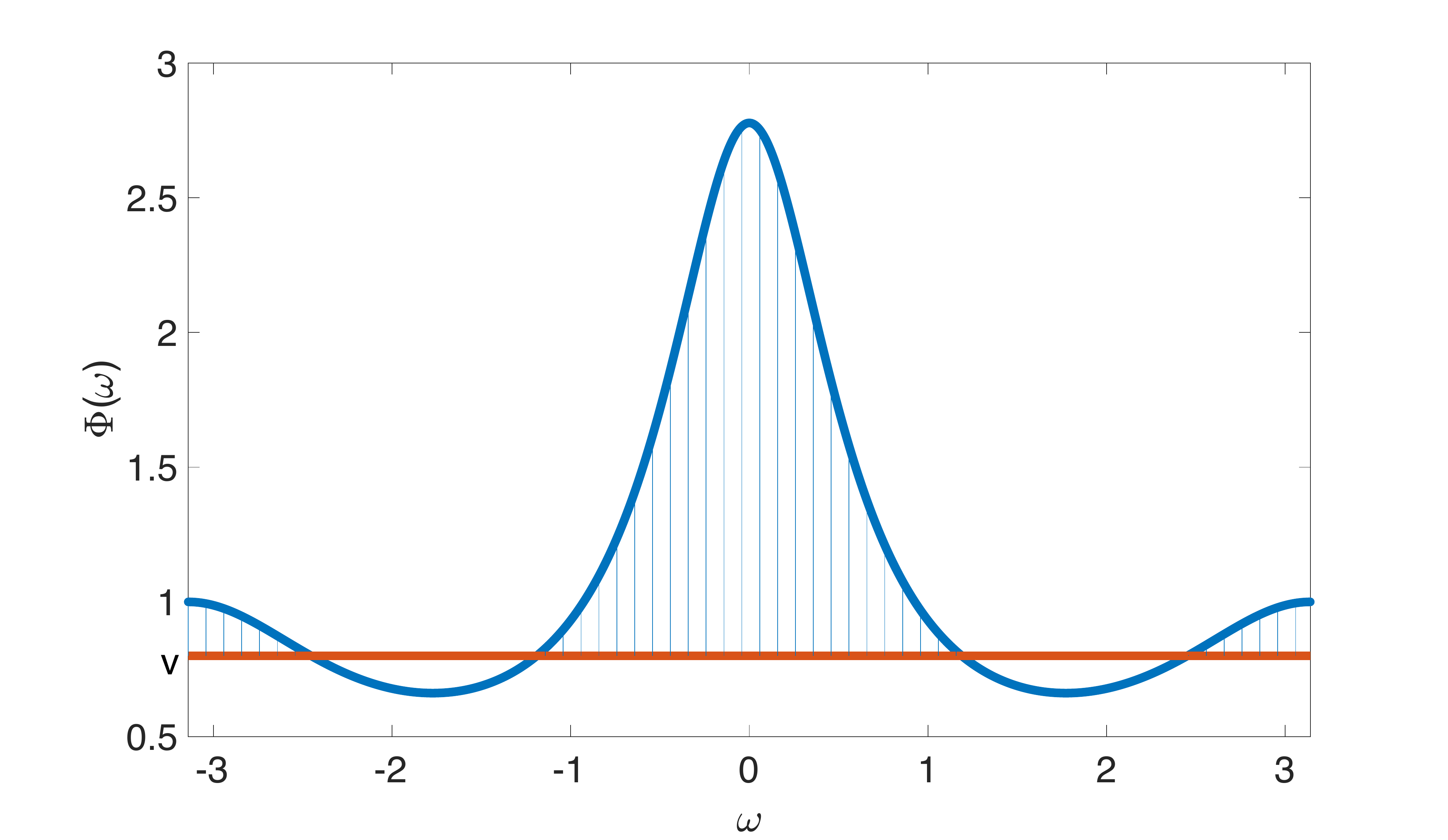}
          	\caption{Reverse water-filling for multivariate Gaussian sources.}
          	\label{fig:waterfilling}
          \end{figure}
          
          The similarity threshold limit for the $D_{\text{ID}}$-{\itshape achievable} rate of Gaussian sources with memory is twice the power of the given signal (\ref{eq: signalpower}). That is, if the given similarity threshold is larger than the twice of signal power, the two signals are inherently similar, and there is no system that can achieve a  vanishing $\Pr\{\text{maybe}\}$. 
          
          In the following example, we plot the identification rate curves for Gauss-Markov processes with different correlation coefficients.
  \begin{exmp}
  	\normalfont
  	We consider zero-mean Gauss-Markov processes with unit variance. The power spectral density of the Gauss-Markov process is 
  	\begin{equation}
  	\Phi(\omega) = \frac{1-\rho^2}{1-2\rho\cos(\omega)+\rho^2}.
  	\end{equation}
  	 The largest value of $\Phi(\omega) $ is obtained when $\cos(2\pi\omega) = 1$. We plot the identification rate function for Gauss-Markov processes with $\rho = 0$, $\rho = 0.5$ and  $\rho = 0.9$ respectively. The integral in (\ref{eq:rwithm}) and (\ref{eq:dwithm}) can be approximated by the Riemann sum. We also plot the identification rate of i.i.d. Gaussian sources for reference. It is overlapped with the $\rho = 0$ case as expected.
  	 
  	 We can also observe that the identification rates approach infinity when the similarity thresholds reach their corresponding limits. Since more correlated signals have higher signal power, their corresponding similarity threshold limits are also larger.
  \end{exmp}
  \begin{figure}[!ht]
  	\centering
  	\includegraphics[width=0.5\textwidth]{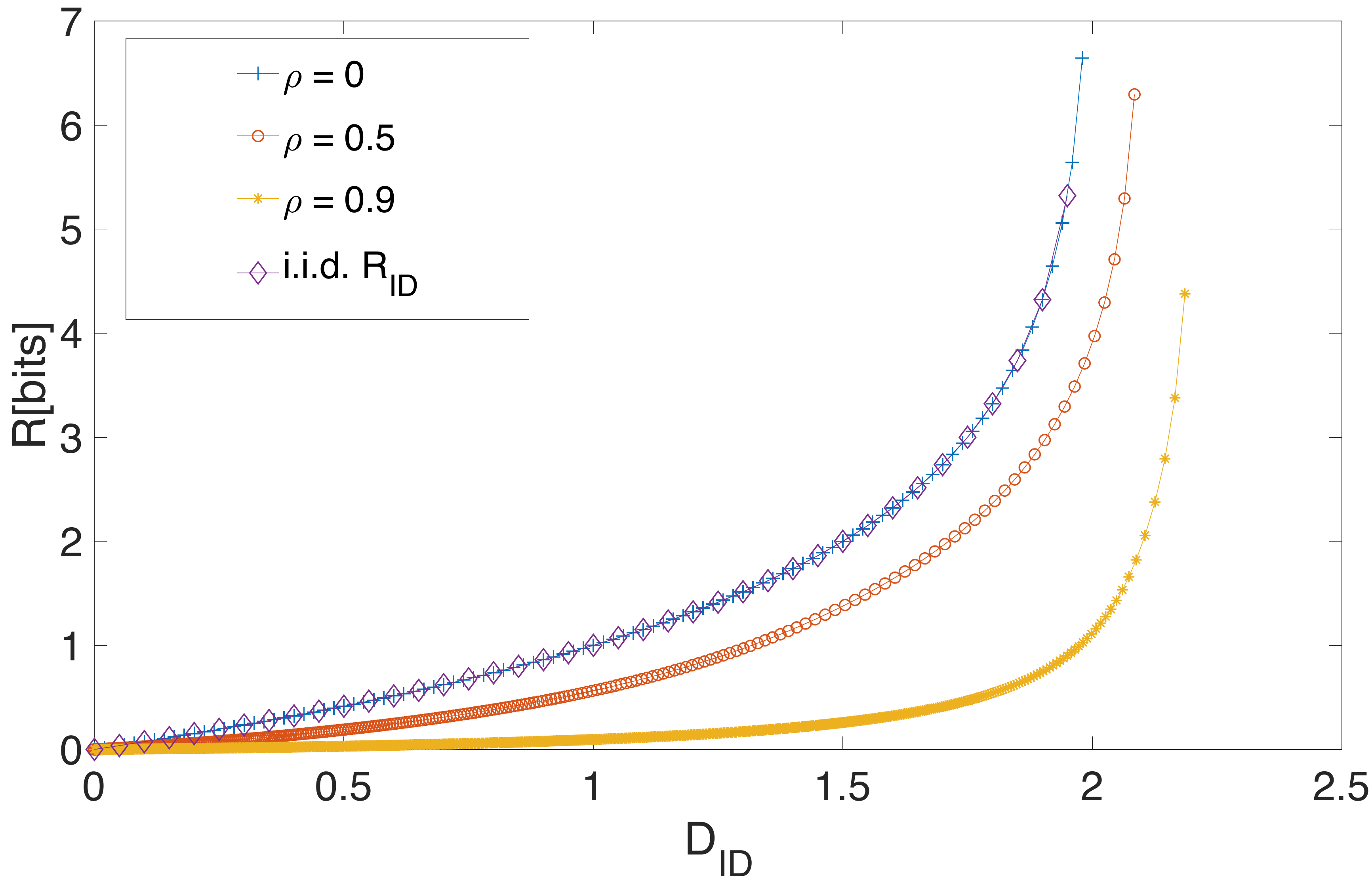}
  	\caption{Comparison of identification rates for Gauss-Markov processes with different correlation coefficients.}
  	\label{fig:mwgfigure}
  \end{figure}
	\section{Component-based model with \\ practical schemes}
	\label{sec:TC-comp}
	 Theorem \ref{theorem:2} is derived on the premise that each component uses an optimal $D_{\text{ID}}$-admissible system. However, the optimal $D_{\text{ID}}$-admissible system is difficult to achieve due to the triangle-inequality constraint that most distortion measures possess. The state-of-the-art practical schemes for the similarity identification problem are the triangle-inequality based TC-$\triangle$ and LC-$\triangle$ schemes proposed in \cite{OCHOA:dcc:14}, where the TC-$\triangle$ scheme is consistently performs better than the LC-$\triangle$ scheme. Therefore, we replace the ideal scheme with the practical TC-$\triangle$ scheme for each component. The described component-based model equipped with TC-$\triangle$ schemes is illustrated in Figure. \ref{fig:tscheme}.}
\begin{figure*}[t!]
\centering
\includegraphics[width=0.8\textwidth]{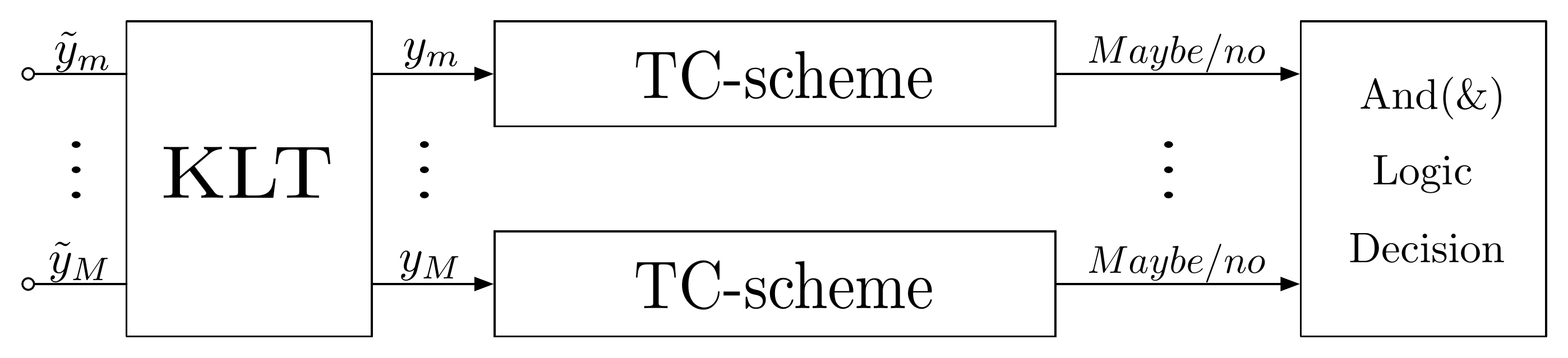}
\caption{Component-based model with TC-$\triangle$ schemes.}
\label{fig:tscheme}
\end{figure*}
 
 We denote the \hanwei{minimum achievable rates} of TC-$\triangle$ and LC-$\triangle$ as $R_{\text{ID}}^{\text{TC}-\triangle}$ and $R_{\text{ID}}^{\text{LC}-\triangle}$, respectively. The authors in \cite{amir:14:isit} show that the minimum achievable rates of TC-$\triangle$ and LC-$\triangle$ schemes generally hold the relation $R_{\text{ID}}^* \leq R_{\text{ID}}^{\text{TC}-\triangle} \leq R_{\text{ID}}^{\text{LC}-\triangle}$. \hanwei{Hence, we select the TC-$\triangle$ scheme as the $D_{\text{ID}}^{(m)}$-admissible system for each component.}
		
The next step is to evaluate the rate-similarity performance of component-based models constructed by TC-$\triangle$ schemes. While the $R_{\text{ID}}^{\text{LC}-\triangle}$ can be evaluated by employing a rate-distortion code on the triangle-inequality principle, the $R_{\text{ID}}^{\text{TC}-\triangle}$ can only be computed  numerically. Here, we propose an iterative method to numerically approximate the minimum achievable rate $R_{\text{ID}}^{\text{TC}-\triangle}$ of TC-$\triangle$ schemes. We only consider the special case where $P_X = P_Y$. The general case that the query and the database are drawn from different distributions can be naturally extended in similar ways as shown in previous works \cite{amir:dcc:13}, \cite{Amir:15:it}.
		\subsection{Iterative Method for Approximating $R_{\text{ID}}^{\text{TC}-\triangle}$}
		It is shown in \cite{amir:14:isit} that any similarity threshold below $D_{\text{ID}}^{\text{TC-}\triangle}(R_{\text{ID}})$ can be attained by a TC-$\triangle$ scheme of rate $R$, where
		\begin{equation}
			\label{eq: tcscheme}
			D_{\text{ID}}^{\text{TC-}\triangle}(R_\text{ID}) \triangleq \underset{P_{\hat{X}|X}: I(X;\hat{X})\leq R_{\text{ID}}}{\max} \mathbb{E}[\rho(\hat{X}, Y)] -\mathbb{E}[\rho(\hat{X}, X)],
		\end{equation}
		where $\hat{X}$ is the reconstructed codeword, and $\hat{X}$ and $Y$ are independent. Since $R_{\text{ID}}^{\text{TC}-\triangle}$ is a strictly increasing function with $D_{\text{ID}}>0$ \cite{amir:14:isit}, for any $s \geq 0$, there exists an exposed point on the $R_{\text{ID}}^{\text{TC}-\triangle}$ curve such that the slope of a tangent to the curve at that point is equal to $s$. 
		
		Denote the exposed points on the $R_{\text{ID}}^{\text{TC}-\triangle}$ curve by $(D_s, R_{\text{ID}}^{\text{TC}-\triangle}(D_s))$, where $D_s$ follows (\ref{eq: tcscheme}):
		\begin{equation}
			D_s = \mathbb{E}[\rho(\hat{X}, Y)] -\mathbb{E}[\rho(\hat{X}, X)].
		\end{equation}
		The achievable rate region is the area above the $(D_s,R_{\text{ID}}^{\text{TC}-\triangle}(D_s))$ curve. To obtain the exposed point on the curve, it is equivalent to minimize the intersection of the tangent of the exposed point with the ordinate.
		\begin{equation}
			\label{eq:loss1}
			\min R_{\text{ID}}^{\text{TC}-\triangle}(D_s)-sD_s.
		\end{equation}
	By varying over all $s \geq 0$, we then trace out the whole rate-similarity curve. 
	
	In the following, we denote the truncated discretized distribution of the source as $\mathbf{P}_{\mathbf{X}} = [p(x_1), p(x_2),..., p(x_n))]^T$ and the marginal distribution of reconstructed codewords as $\mathbf{t}_{\hat{\mathbf{x}}} = [t(\hat{x}_1), t(\hat{x}_2),..., t(\hat{x}_m)]^T$. We form the conditional probability mass functions as columns of an $m \times n$ matrix:
	\[\mathbf{Q} =
	\begin{bmatrix}
	Q(\hat{x}_1|x_1)  & \dots  & Q(\hat{x}_1|x_n) \\
	\vdots &  \ddots & \vdots \\
	Q(\hat{x}_m|x_1)  & \dots  & Q(\hat{x}_m|x_n)
	\end{bmatrix}
	\]
	
	The expected distortions between $X$ and $\hat{X}$ when averaging over their marginal distributions and joint distribution are
	\begin{equation}
		\mathbb{E}_{P_X \times P_{\hat{X}}}\rho(X, \hat{X}) = \sum_{i =1}^{n}\sum_{j = 1}^mp(x_i)t(\hat{x}_j)\rho(x_i,\hat{x}_j), 
	\end{equation} 
	and 
	\begin{equation}
		\mathbb{E}_{P_{X,\hat{X}}}\rho(X, \hat{X}) = \sum_{i =1}^{n}\sum_{j = 1}^mp(x_i)Q(\hat{x}_j|x_i)\rho(x_i,\hat{x}_j).
	\end{equation} 
	Since $R_{\text{ID}}\leq\inf_{\mathbf{Q}} I(X;\hat{X})$, the objective function (\ref{eq:loss1}) can be expressed as a minimization over $\mathbf{Q}$
	\begin{align}
		\label{equ:obj}
		&R_{\text{ID}}^{\text{TC}-\triangle}(D_s)-sD_s \\
		=& \underset{\mathbf{Q}}{\min}(I(\mathbf{P}_{\mathbf{X}},\mathbf{Q}) -sD_s) \nonumber \\
		\label{eq:fullordinate}
		\geq &\underset{\mathbf{Q}}{\min}[I(\mathbf{P}_{\mathbf{X}},\mathbf{Q})-s(\mathbb{E}_{P_X \times P_{\hat{X}}}\rho(X, \hat{X}) -\mathbb{E}_{P_{X,\hat{X}}}\rho(X, \hat{X})].
	\end{align}
	Note that the elements of $\mathbf{Q}$ represent probabilities and each column of $\mathbf{Q}$ is a probability mass function. This introduces the constraints $Q(\hat{x}_j|x_i) \geq 0$ and $\sum_{j = 1}^mQ(\hat{x}_j|x_i) = 1$. We temporarily ignore the constraint $Q(\hat{x}_j|x_i) \geq 0$ and define a Lagrange cost function as
	\begin{equation}
		\label{equ:cost1}
		\begin{split}
			J(Q) &= \sum_{i =1}^{n}\sum_{j = 1}^m p(x_i)Q(\hat{x}_j|x_i)\log\left(\frac{Q(\hat{x}_j|x_i)}{t(\hat{x}_j)}\right) -\\
			&s \sum_{i =1}^{n}\sum_{j = 1}^m p(x_i) \rho(x_i,\hat{x}_j)(t(\hat{x}_j)-Q(\hat{x}_j|x_i)) + \\
			& \sum_{i = 1}^n v_i \sum_{j = 1}^mQ(\hat{x}_j|x_i),
		\end{split}
	\end{equation}
	where $v_i$ are the Lagrange multipliers.
	
	Differentiating with respect to $Q(\hat{x}_j|x_i)$, we have
	\begin{align}
		\label{eq:ba1}
		\frac{\partial J}{\partial Q(\hat{x}_j|x_i)} &= p(x_i)\log\frac{Q(\hat{x}_j|x_i)}{t(\hat{x}_j)} + sp(x_i)(\rho(x_i,\hat{x}_j)\nonumber\\&-p(x_i)\sum_{x'}p(x')\rho(x', \hat{x}_j)) + v_i = 0.
	\end{align}
	Setting $\log \mu(x_i) = \frac{v_i}{p(x_i)}$, we obtain from (\ref{eq:ba1})
	\begin{align}
		&p(x_i)\left[\log\frac{Q(\hat{x}_j|x_i)}{t(\hat{x}_j)}+s(\rho(x_i,\hat{x}_j)-\right.\nonumber\\&\left.p(x_i)\sum_{x'}p(x')\rho(x', \hat{x}_j)) + \log \mu(x_i)\right] = 0,
	\end{align}
	or 
	\begin{equation}
		\label{equ:optimize2}
		Q(\hat{x}_j|x_i) = \frac{t(\hat{x}_j) e^{-s(\rho(x_i,\hat{x}_j)-p(x_i)\sum_{x'}p(x')\rho(x', \hat{x}_j))}}{\mu(x_i)}.
	\end{equation}
	Since $\sum_{j = 1}^{m}Q(\hat{x}_j|x_i) = 1$, we have
	\begin{equation}
		\mu(x_i) = \sum_{j = 1}^{m} t(\hat{x}_j) e^{-s(\rho(x_i,\hat{x}_j)-p(x_i)\sum_{x'}p(x')\rho(x', \hat{x}_j))}.
	\end{equation}
	We can see that $Q(\hat{x}_j|x_i)$  is always nonnegative. 
	
	To vectorize the above operations, we define the $m \times n$ distortion matrix as:
		\[ \mathbf{\Gamma} =
		\begin{bmatrix}
		\rho(x_1,\hat{x}_1)  & \dots  & \rho(x_n,\hat{x}_1) \\
		\vdots  &  \ddots & \vdots \\
		\rho(x_1,\hat{x}_m)  & \dots  & \rho(x_n,\hat{x}_m)
		\end{bmatrix}
		\] Hence, the conditional probability mass function (\ref{equ:optimize2}) can be expressed as 
	\begin{equation}
		\label{eq:iterupdate2}
		\mathbf{Q}' = \exp(-s(\mathbf{\Gamma}-\mathbf{\Gamma}\mathbf{P}_{\mathbf{X}}\mathbf{P}_{\mathbf{X}}^T))\odot\mathbf{t}_{\hat{\mathbf{X}}},
	\end{equation}
	where $\odot$ stands for column-wise multiplication. Then $\mathbf{Q}' $ can be further normalized by
	\begin{equation}
			\label{eq:iterupdate3}
		\mathbf{Q} = \mathbf{Q}' \oslash (\mathbf{Q}'\mathbf{e}_n),
	\end{equation}
	where $\mathbf{e}_n = (1,1,...,1)^T$, and where $\oslash$ denotes the row-wise division.
	
	The conditional probability matrix $\mathbf{Q}$ from (\ref{equ:optimize2}) can be expressed analytically if we assume that $t_\mathbf{\hat{x}}$ is known. In our iterative method, we first initialize $\mathbf{t}_{\hat{\mathbf{X}}}$ and choose an $s$. Then, $\mathbf{Q}$ is determined according to (\ref{equ:optimize2}). The marginal codeword distribution is updated by the Bayes' rule
		\begin{equation}
			t^*(\hat{x}_j) = \sum_{i = 1}^n p(x_i)Q(\hat{x}_j|x_i).
		\end{equation}
	The corresponding vectorized representation is 
	\begin{equation}
		\label{equ:optimize1}
		\mathbf{t}_{\hat{x}}^*= \mathbf{QP_X}.
	\end{equation}
	
	We update $\mathbf{Q}$ and $\mathbf{t}_{\hat{\mathbf{X}}}$ according to (\ref{eq:iterupdate2}, \ref{eq:iterupdate3}) and (\ref{equ:optimize1}) iteratively until the algorithm converges. Note that the term $\mathbf{\Gamma} \mathbf{P}_{\mathbf{X}} \mathbf{P}_{\mathbf{X}}^T$ in (\ref{eq:iterupdate2}) is constant, so it can be computed before the iterations.  Finally, we approximate one point of $R_{\text{ID}}^{\text{TC}-\triangle}$ by evaluating $I(\mathbf{P}_{\mathbf{X}}, \mathbf{Q})$ and $D_s$. 
	\begin{exmp}
		\normalfont
		We verify our algorithm by testing it on the binary-Hamming case. Figure \ref{fig:binary} shows that the minimum achievable rates as computed by using the derived algorithm for TC-$\triangle$ schemes is the same as the identification rate $R_{\text{ID}}^*$ of binary sources with Hamming distance. This is consistent with the special case that $R_{\text{ID}}(D_\text{ID}) = R_{\text{ID}}^{\text{TC-}\triangle}(D_\text{ID})$ for the binary-Hamming case.
	\end{exmp}
	\begin{figure}[!ht]
		\centering
		\includegraphics[width=0.5\textwidth]{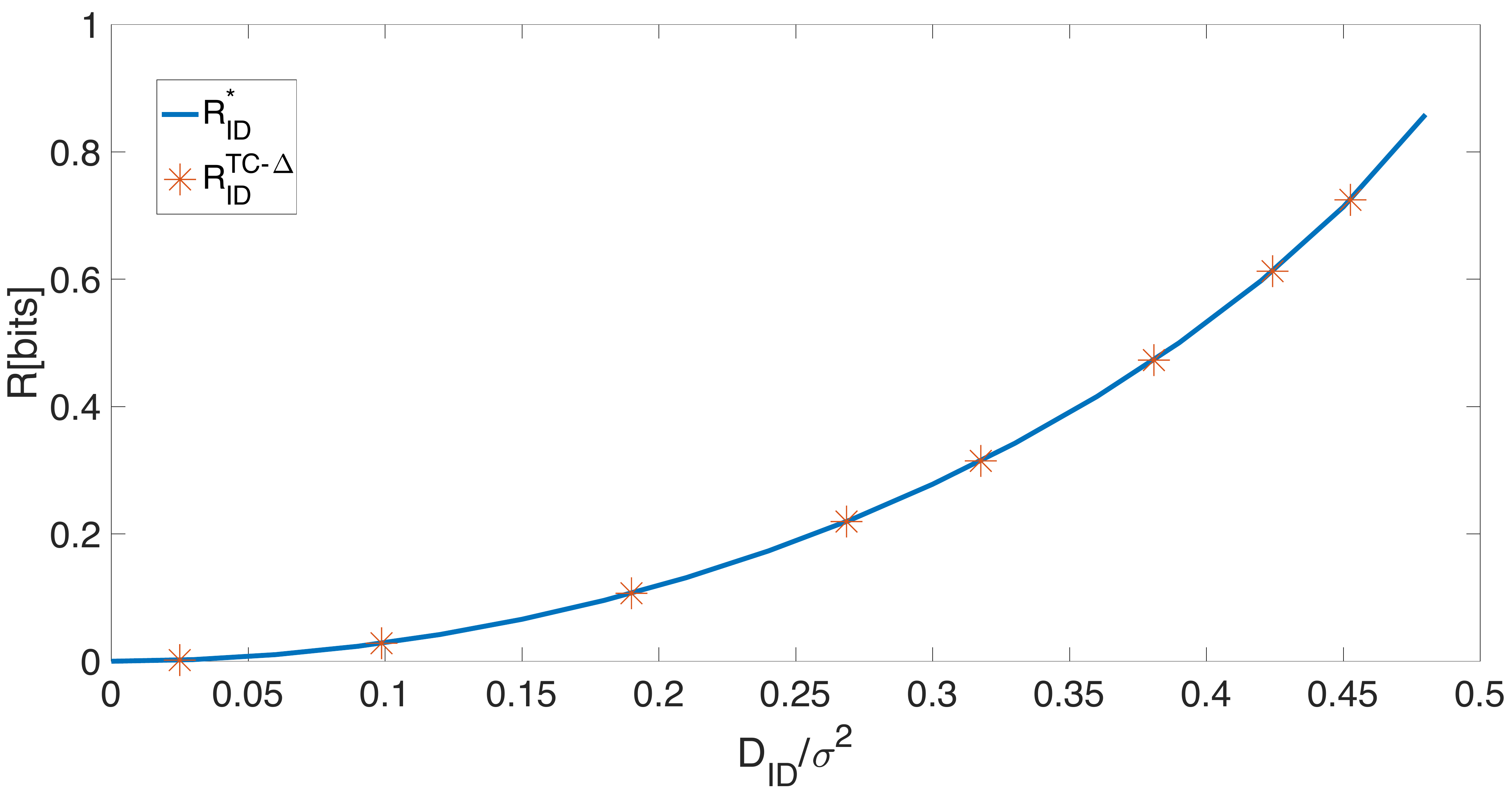}
		\caption{Binary source with Hamming distortion: $P_X = P_Y = \text{Bern}(0.5)$.}
		\label{fig:binary}
	\end{figure}
	\hanwei{
	\subsection{Iterative Method for Component-based Model} 
	\label{sec:comp-tc}
	The optimal rate allocation of the component-based model constructed by $\text{TC-}\triangle$ schemes can be simply achieved by applying the Pareto condition. That is, each component system should operate at the point where all rate-similarity curves of the components have the same slopes
     \begin{equation}
     \frac{\partial J}{\partial D_s^{(i)}} = \frac{\partial J}{\partial D_s^{(j)}} = s, 
     \end{equation}
     where $i, j \in [1,M]$, $s$ is the chosen value of the slope. Then we can obtain the rate-similarity curve of the $M$-component model by traversing the values of $s \in [0, \infty)$ with each component running the iterative method of the TC-$\triangle$ scheme independently for a given $s$.}
 
	\subsection{Comparisons} 
	\label{sec:comparisions}
	In this section, we use the proposed iterative method to approximate the $R_{\text{ID}}^{\text{TC-}\triangle}$ for both i.i.d. and multivariate Gaussian sources, and then compare them with $R_{\text{ID}}^{\text{LC-}\triangle}$ and $R_{\text{ID}}^*$ of optimal schemes. First, we derive the $R_{\text{ID}}^{\text{LC-}\triangle}$ of LC-$\triangle$ schemes for quadratic Gaussian sources by employing a rate-distortion code on the triangle-inequality principle (\ref{equ:tri1}) \cite{amir:dcc:13}. 
	\begin{equation}
		\label{equ:tri1}
		\sqrt{\mathbb{E}_{P_X \times P_{\hat{X}}} \left[(X-\hat{X})^2\right]} \geq \sqrt{\mathbb{E}_{P_{X,\hat{X}}}\left[(X-\hat{X})^2\right]} + \sqrt{D_{\text{ID}}},
	\end{equation}
	Consider a Gaussian source $X\sim \mathcal{N}(0, \sigma^2)$ that is compressed by an optimal rate-distortion code with $(R, D)$. The rate-distortion code can be designed with the codeword distribution $\hat{X}\sim \mathcal{N}(0, \sigma^2-D)$ \cite{Cover2006}, and we have $\mathbb{E}_{P_{X,\hat{X}}}(X-\hat{X})^2 = D$. Hence, from (\ref{equ:tri1}), we can obtain the relation between quantization distortion $D$ and similarity threshold $D_{\text{ID}}$ as
	\begin{equation}
		0 \leq D \leq \frac{1}{2}(2\sigma^2 - \sqrt{D_{\text{ID}}(4\sigma^2-D_{\text{ID}})}).
	\end{equation}
	
	Then, the rate we need to compress the Gaussian source under the maximum distortion constraint is equal to the minimum required $R_{\text{ID}}$ for the LC-$\triangle$ scheme
		\begin{equation}
			\label{eq:ridlc}
			R_{\text{ID}}^{\text{LC-}\triangle}(D_\text{ID}) = \frac{1}{2}\log \frac{\sigma^2}{D} =\frac{1}{2}\log(\frac{1}{1-\sqrt{2\frac{D_{\text{ID}}}{2\sigma^2}-(\frac{D_{\text{ID}}}{2\sigma^2})^2}}).
		\end{equation}
		
		In Figure \ref{fig:rd}, we compare the $R_{\text{ID}}^{\text{LC-}\triangle}$ with the $R_{\text{ID}}^{\text{TC}-\triangle}$ approximation $R_{\text{ID}}^I$ and the $R_{\text{ID}}^*$ for i.i.d. Gaussian sources. It shows that both $\text{LC-}\triangle$ and $\text{TC-}\triangle$ schemes are suboptimal for quadratic Gaussian queries. Although the gap between $R_{\text{ID}}^{\text{LC-}\triangle}$ and $R_{\text{ID}}^I$ is small, we see from the zoomed plot that the $R_{\text{ID}}^I$ is constantly lower than the $R_{\text{ID}}^{\text{LC-}\triangle}$.
		
		\hanwei{In Figure \ref{fig:rd1}, we compare $R_{\text{ID}}^{\text{LC-}\triangle}$, $R_{\text{ID}}^I$, $R_{\text{ID}}^{M*}$ and the approximate minimum achievable rate $R_{\text{ID}}^{\text{IC}}$ of the component-based model constructed by $\text{TC-}\triangle$ schemes for multivariate Gaussian sources with $\rho = 0.7$}. Since there is no simple analytical rate-distortion function for correlated Gaussian sources, the $R_{\text{ID}}^{\text{LC-}\triangle}$ is evaluated by computing a standard BA algorithm. Figure \ref{fig:rd1} shows that both $\text{LC-}\triangle$ and $\text{TC-}\triangle$ (the latter as approximated by $R_{\text{ID}}^I$) are suboptimal when compared to the component-based model with optimal $D_{\text{ID}}^{(m)}$-admissible components. Again, the gap between $R_{\text{ID}}^{\text{LC-}\triangle}$ and $R_{\text{ID}}^I$ is small, and the approximated TC-$\triangle$ scheme constantly performs better than $\text{LC-}\triangle$. The small gap between $R_{\text{ID}}^{\text{LC-}\triangle}$ and $R_{\text{ID}}^{\text{I}}$ results from the codeword update step (\ref{eq:iterupdate2}). Compared to the standard BA algorithm, the proposed iterative method has a constant shift $\mathbf{\Gamma} \mathbf{P}_{\mathbf{X}} \mathbf{P}_{\mathbf{X}}^T$ in the exponential function.
		
		\hanwei{The $R_{\text{ID}}^{\text{IC}}$ can be obtained by the Pareto condition described in Section \ref{sec:TC-comp}.\ref{sec:comp-tc} It shows that the component-based model constructed by $\text{TC-}\triangle$ schemes can achieve better performance than the $\text{TC-}\triangle$ and $\text{LC-}\triangle$ schemes unaided for multivariate Gaussian sources. We should note that this is not a contradiction to the classical rate distortion quantizers where vector quantizers usually perform better due to the space-filling advantage \cite{looka:89:it}. The better performance of the component-based model follows from the reason that the overall similarity threshold is achieved by a specific optimal rate-similarity allocation among components.}
	\begin{figure}[!ht]
		\centering
		\includegraphics[width=0.5\textwidth]{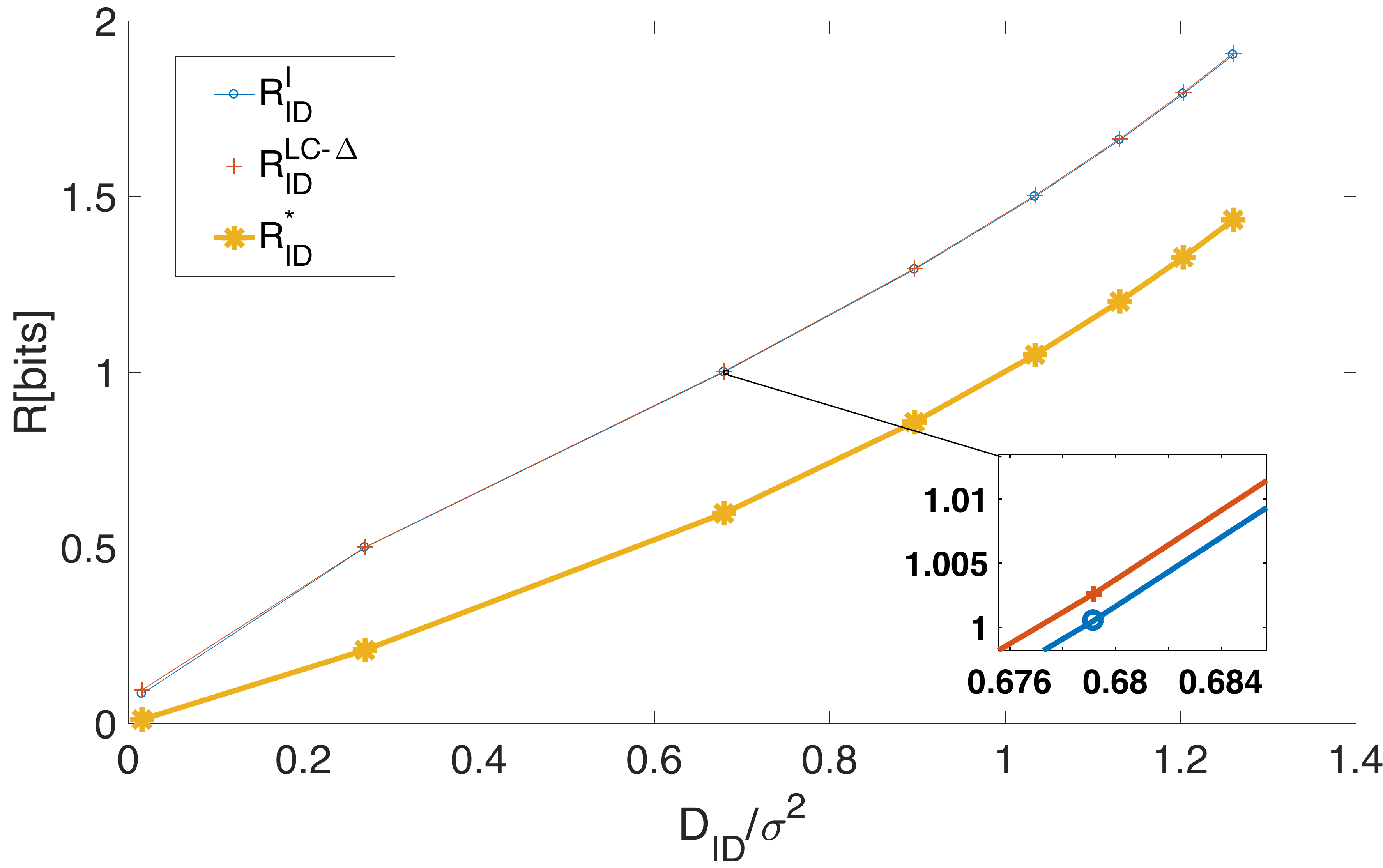}
		\caption{Comparison of $R_{\text{ID}}^{\text{LC-}\triangle}$, $R_{\text{ID}}^I$ and $R_{\text{ID}}^*$ for i.i.d. Gaussian sources.}
		\label{fig:rd}
	\end{figure}
	\begin{figure}[!ht]
		\centering
		\includegraphics[width=0.5\textwidth]{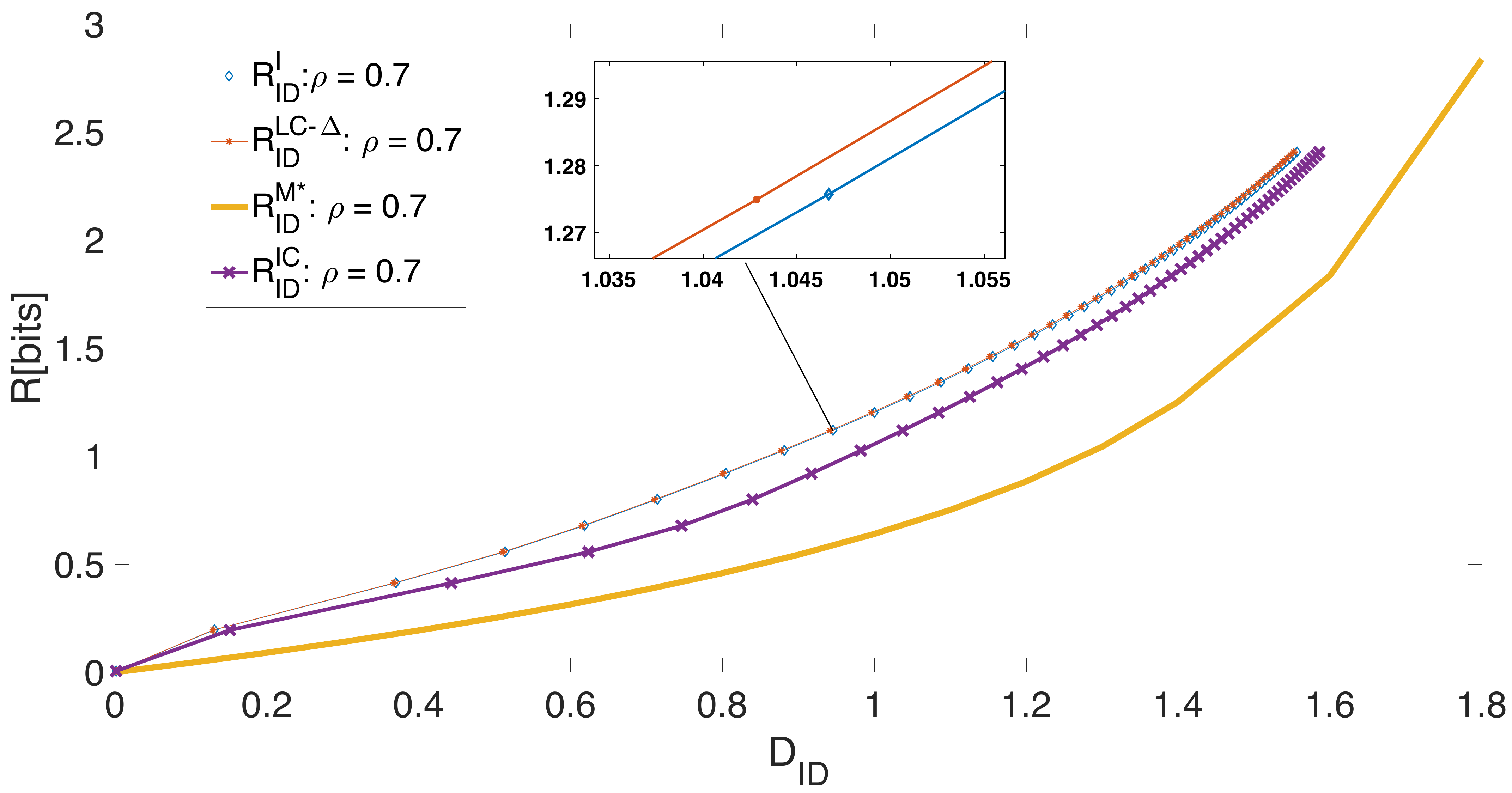}
		\caption{Comparison of $R_{\text{ID}}^{\text{LC-}\triangle}$, $R_{\text{ID}}^I$, $R_{\text{ID}}^{C*}$ and $R_{\text{ID}}^{\text{IC}}$ for multivariate Gaussian sources with $\rho = 0.7$.}
		\label{fig:rd1}
	\end{figure}
	\section{Conclusions}
	In this work, we study the component-based models for correlated similarity queries. We show that the component-based model with KLT transform under the optimal rate-similarity allocation can perfectly represent multivariate Gaussian queries. Hence, we can derive its identification rate based on the model. We then extend the identification rate result of multivariate Gaussian signals to general Gaussian sources with memory and show that it follows a "reverse water-filling" process. Furthermore, we evaluate the performance of the component-based model constructed by TC-$\triangle$ schemes. The simulation shows that our component-based model with TC-$\triangle$ schemes can achieve better performance than the TC-$\triangle$ schemes unaided for multivariate Gaussian sources.
	\bibliographystyle{IEEEtran}
	\bibliography{fine3}

\begin{thebibliography}{10}
\providecommand{\url}[1]{#1}
\csname url@samestyle\endcsname
\providecommand{\newblock}{\relax}
\providecommand{\bibinfo}[2]{#2}
\providecommand{\BIBentrySTDinterwordspacing}{\spaceskip=0pt\relax}
\providecommand{\BIBentryALTinterwordstretchfactor}{4}
\providecommand{\BIBentryALTinterwordspacing}{\spaceskip=\fontdimen2\font plus
\BIBentryALTinterwordstretchfactor\fontdimen3\font minus
  \fontdimen4\font\relax}
\providecommand{\BIBforeignlanguage}[2]{{%
\expandafter\ifx\csname l@#1\endcsname\relax
\typeout{** WARNING: IEEEtran.bst: No hyphenation pattern has been}%
\typeout{** loaded for the language `#1'. Using the pattern for}%
\typeout{** the default language instead.}%
\else
\language=\csname l@#1\endcsname
\fi
#2}}
\providecommand{\BIBdecl}{\relax}
\BIBdecl

\bibitem{amir:14:isit}
A.~Ingber and T.~Weissman, ``Compression for similarity identification:
  Fundamental limits,'' in \emph{IEEE International Symposium on Information
  Theory}, Jun. 2014, pp. 1--5.

\bibitem{ahlswede:99:it}
R.~Ahlswede, E.~H. Yang, and Z.~Zhang, ``Identification via compressed data,''
  \emph{{IEEE} Trans. Inf. Theory}, vol.~43, no.~1, pp. 48--70, 1997.

\bibitem{Amir:15:it}
A.~Ingber, T.~Courtade, and T.~Weissman, ``Compression for quadratic similarity
  queries,'' \emph{{IEEE} Trans. Inf. Theory}, vol.~61, no.~5, pp. 2729 --2747,
  May 2015.

\bibitem{amir13}
A.~Ingber and T.~Weissman, ``The minimal compression rate for similarity
  identification,'' [Online]. Available: http://arxiv.org/abs/1312.2063.

\bibitem{OCHOA:dcc:14}
I.~Ochoa, A.~Ingber, and T.Weissman, ``Compression schemes for similarity
  queries,'' in \emph{Proc. of the IEEE Data Compression Conference}, Mar.
  2014.

\bibitem{steiner:16:it}
F.~Steiner, S.~Dempfle, A.~Ingber, and T.~Weissman, ``Compression for quadratic
  similarity queries: finite blocklength and practical schemes,'' \emph{{IEEE}
  Trans. Inf. Theory}, vol.~62, no.~5, pp. 2737 --2747, May 2016.

\bibitem{hamkins:97:it}
J.~Hamkins and K.~Zeger, ``Asymptotically dense spherical codes. i. wrapped
  spherical codes,'' \emph{{IEEE} Trans. Inf. Theory}, vol.~43, no.~6, pp.
  1774--1785, 1997.

\bibitem{hanwei:dcc:17}
H.~Wu, Q.~Wang, and M.~Flierl, ``Tree-structured vector quantization for
  similarity queries,'' in \emph{2017 Data Compression Conference (DCC)}, Apr.
  2017.

\bibitem{hanwei:17:asilomar}
H.~Wu and M.~Flierl, ``Transform-based compression for quadratic similarity
  queries,'' in \emph{Conference on Signals, Systems, and Computers}, Oct.
  2017, pp. 377--381.

\bibitem{blahut:72:it}
R.~E. Blahut, ``Computation of channel capacity and rate distortion
  functions,'' \emph{{IEEE} Trans. Inf. Theory}, vol. IT-18, no.~4, pp.
  460--473, 1972.

\bibitem{arimoto:72:it}
S.~Arimoto, ``An algorithm for computing the capacity of arbitrary memoryless
  channels,'' \emph{{IEEE} Trans. Inf. Theory}, vol. IT-18, no.~1, pp. 14--20,
  1972.

\bibitem{walsh:15:sp}
G.~Ku, J.~Ren, and J.~M. Walsh, ``Computing the rate distortion region for the
  ceo problem with independent sources,'' \emph{{IEEE} Trans. Signal Process.},
  vol.~63, no.~3, pp. 567 --575, Feb. 2014.

\bibitem{kraskov2004estimating}
A.~Kraskov, H.~St{\"o}gbauer, and P.~Grassberger, ``Estimating mutual
  information,'' \emph{Physical review E}, vol.~69, no.~6, p. 066138, 2004.

\bibitem{greander58}
U.~Grenander and G.~Szegö, \emph{Toeplitz Forms and Their Applications}.\hskip
  1em plus 0.5em minus 0.4em\relax Berkeley and Los Angeles, USA: University of
  California Press, 1958.

\bibitem{gray2006}
R.~M. Gray, \emph{Toeplitz and Circulant Matrices: A Review}.\hskip 1em plus
  0.5em minus 0.4em\relax Delft, The Netherlands: Now publishers Inc, 2006.

\bibitem{Cover2006}
T.~M. Cover and J.~A. Thomas, \emph{Elements of Information Theory},
  2nd~ed.\hskip 1em plus 0.5em minus 0.4em\relax New York, NY, USA: Wiley,
  2006.

\bibitem{amir:dcc:13}
A.~Ingber, T.~Courtade, and T.Weissman, ``Quadratic similarity queries on
  compressed data,'' in \emph{2013 Data Compression Conference (DCC)}, Apr.
  2013.

\bibitem{looka:89:it}
T.~Lookabaugh and R.~Gray, ``High-resolution quantization theory and the vector
  quantizer advantage,'' \emph{{IEEE} Trans. Inf. Theory}, vol.~35, no.~5, pp.
  1020--1033, Sep. 1989.

\end{thebibliography}
\end{document}